\documentclass{article}
\usepackage{amsmath, amssymb, amsthm}
\usepackage{setspace}
\usepackage[letterpaper, margin = 1 in]{geometry}
\usepackage{authblk}
\newtheorem{theorem}{Theorem}[section]
\newtheorem{lemma}[theorem]{Lemma}

\theoremstyle{remark}
\newtheorem*{remark}{Remark}

\newtheorem*{acknowledgment}{Acknowledgment}
\numberwithin{equation}{section}
\numberwithin{theorem}{section}
\DeclareMathOperator*{\esssup}{ess\;sup}
\DeclareSymbolFont{bbold}{U}{bbold}{m}{n}
\DeclareSymbolFontAlphabet{\mathbbold}{bbold}

\begin{document}
\title{Estimates on Functional Integrals of Quantum Mechanics and Non-Relativistic Quantum Field Theory}
\author[1]{Gonzalo A. Bley \thanks{gb3kd@virginia.edu}}
\author[1]{Lawrence E. Thomas \thanks{let@virginia.edu}}
\affil[1]{Department of Mathematics, University of Virginia,\authorcr Charlottesville, VA 22904, USA}

\maketitle


\begin{abstract}
We provide a unified method for obtaining upper bounds for certain functional integrals appearing in quantum mechanics 
and non-relativistic quantum field theory, functionals of the form $E\left[\exp(\mathcal{A}_T)\right]$, the (effective) action 
$\mathcal{A}_T$ being a function of particle trajectories up to time $T$. The estimates in turn yield rigorous lower bounds for ground state energies, 
via the Feynman-Kac formula. The upper bounds are obtained by writing the action for these functional integrals in terms of stochastic integrals. 
The method is illustrated in familiar quantum mechanical settings: for the hydrogen atom, for a Schr\"{o}dinger operator with $1/|x|^2$ 
potential with small coupling, and, with a modest adaptation of the method, for the harmonic oscillator. We then present our principal 
applications of the method, in the settings of non-relativistic quantum field theories for particles moving in a quantized Bose field, 
including the optical polaron and Nelson models.
\end{abstract}


\begin{section}{Introduction}
The intent of this article is to provide simple upper bounds on moment-generating functions arising
from problems in elementary quantum mechanics and non-relativistic quantum field theories. The bounds
are applicable to Feynman-Kac expectations $E^x\left[\exp(-\int_0^T V(s, X_s)ds)\right]$ in quantum mechanics,
with $E^x[\cdot]$ expectation with respect to $d$-dimensional Brownian motion $X_s$ starting at $x$. ($E^0$ will be written as $E$.)
 The bounds are also applicable to 
 problems involving non-relativistic particles interacting with a quantized Boson field.  After integrating out 
the field variables, one typically encounters expectations of the sort
\begin{equation}
E^x\left[\exp\left(\int_0^T\!\!\!\int_0^t f(t-s)/|X^{(i)}_ t- X^{(j)}_s |^{\theta}\,ds\,dt\right)\right],
\label{general_action}
\end{equation}
 in which particle ${i}$  interacts with its past trajectory or that of another particle ${j}$. The bounds are valid for
{\it all} times $T\geq 0$, including $T\rightarrow\infty$.  It will be important to us  in applications that the bounds be log-linear
in time $T$, $T\rightarrow\infty$, in order that energy expressions $-\lim_{T \to \infty}\frac{1}{T}\ln(E[\cdot])$ be finite.

The basic strategy for obtaining these bounds in both quantum mechanical and quantum field theoretic applications is twofold. First, one writes the 
argument of the exponential, i.e., the {\it action}, as its 
expectation plus a ``fluctuating" part, the latter represented by a stochastic integral. Provided the action is an $L^2$-functional, it can be
decomposed in this manner, with the integrand of the stochastic integral determined with the aid of the Clark-Ocone formula. Second, one then uses a martingale argument  to bound the moment 
generating functional for the stochastic integral, typically in terms of a new action less singular than that of the initial problem.  We do not claim that 
the strategy is altogether new; but we believe that a succinct description of it, together with an investigation into the efficacy of the resulting bounds in a wide 
range of applications, is useful.

Perhaps the simplest illustration of the applicability of these bounds is the case of the hydrogen problem with Hamiltonian
\begin{equation}
H_{\alpha} = -\frac{\Delta}{2} - \frac{\alpha}{|x|}
\end{equation}
acting in $L^2({\mathbb R}^3)$.
The method shows that the Feynman-Kac expression 
\begin{equation}
E\left[\exp\left(\int_0^T\frac{\alpha}{|X_t|}\,dt\right)\right]
\end{equation}
is bounded above by $\exp(2\sqrt{2/\pi}\alpha T^{1/2} + \alpha^2 T/2)$ for all $\alpha \geq 0$ and $T \geq 0$. 
We note that this expression has the correct asymptotic behavior for small $T$, and that it 
provides a lower bound for the ground state energy by Feynman-Kac.  This lower bound is sharp:
\begin{equation}
\text{inf spec }H_\alpha = -\frac{\alpha^2}{2} = -\lim_{T \to \infty}\frac{1}{T}\ln \left(E\left[\exp\left(\int_0^T\frac{\alpha}{|X_t|}\,dt\right)\right]\right)
\end{equation}
(see section 3.1). 
The method accommodates Hamiltonians with potentials having different powers of $|x|$ including the critical case
$V= -\alpha |x|^{-2}$ in $d\geq 3$ dimensions, $\alpha$ suitably small, as well as  situations where $\alpha$ is replaced by a time-dependent  function  (see section 3.2).  

Another  elementary example to be discussed is the use of a stochastic integral representation applied to the harmonic oscillator of quantum mechanics. The result is  another derivation of  a Cameron-Martin formula, meaning here an exact expression for the functional integral yielding the ground state energy of the harmonic 
oscillator \cite{CM} (see also \cite{KS1, RY}). Rather than simply writing the action as its expectation plus a stochastic integral however, 
the strategy is to solve a Hamilton-Jacobi-like equation for the action appearing in the functional integral,  the action still involving a stochastic integral.    Discussion of the harmonic oscillator has been placed at the
end of the article so that its analysis may be contrasted with the method described above (see the appendix, section A.3). We have included this example with its adaptation, suggesting a way of sharpening the method in specific applications.

To illustrate the method as it applies in non-relativistic quantum field theory, we consider the case of the optical polaron model of 
H. Fr\"{o}hlich \cite{HF}, 
a simple model of a non-relativistic electron interacting with the phonons of an ionic crystal. Its Hamiltonian takes the form
\begin{gather}\label{Frohham}
H_{\alpha} \equiv \frac{p^2}{2} + \int_{k \in {\mathbb R}^3} a_k^{\dagger}a_k\,dk + \frac{\sqrt{\alpha}}{2^{3/4}\pi}\int_{k \in {\mathbb R}^3}\frac{1}{|k|}(a_k e^{ikx} + a_k^{\dagger}e^{-ikx})\,dk,
\end{gather}
where $\alpha$ is a dimensionless constant, $x$ is in $\mathbb{R}^3$, $p = -i\nabla$, and the $a_k$ and $a_k^{\dagger}$ are Boson annihilation and 
creation operators that satisfy $[a_k, a_{k^{'}}^{\dagger}] = \delta(k - k^{'})$. Of interest to us will be Feynman's functional integral \cite{RF}
\begin{equation}\label{RF.eq1}
E\left[\exp\left(\frac{\alpha}{\sqrt{2}}\int_0^T\!\!\!\int_0^t\frac{e^{-(t - s)}}{|X_t - X_s|}\,ds\,dt\right)\right],
\end{equation}
which is of the form \eqref{general_action} with $f(t) = \alpha e^{-t}/\sqrt{2}$ and $\theta$ equal to 1. This expression is equal to 
the formal semigroup matrix element $(\delta_x\otimes \Omega_0, e^{-t H_{\alpha}}\mathbbold{1}\otimes \Omega_0)$, where $\Omega_0$ 
is the ground state of the Boson field and $\mathbbold{1}$ is the function identically equal to 1 on $ {\mathbb R}^3$; \eqref{RF.eq1} is obtained from
the matrix element by integrating out the Bose field variables. The limit 
\begin{equation}
-\lim_{T \to \infty}\frac{1}{T}\ln\left[(\delta_x\otimes \Omega_0, e^{-t H_{\alpha}}\mathbbold{1}\otimes \Omega_0)\right]
\end{equation}
then yields the ground state energy of the polaron. We obtain $-\alpha - \alpha^2/4$ as a lower bound for
this energy, which is correct to leading order in $\alpha$ 
for small $\alpha$ and is off from the known behavior $\sim -0.109 \alpha^2$ \cite{DV, LT} 
by a factor of about 2.5 for large $\alpha$. The method is applicable to the multipolaron case as well, see 
\cite{BB1, BB2, FLST}.
See section 3.3 for more details.

Finally, we consider an application of our method to analyzing the Feynman-Kac expression
\begin{equation}\label{Nelson.eq1}
E\left[\exp\left(\alpha\sum_{m, n}\int_0^T\!\!\!\int_0^t\!\!\int_{k \in {\mathbb R}^3}\frac{1}{\omega(k)}e^{ik(X_t^m - X_s^n)}
e^{-\omega(k)(t - s)}\chi_{\Lambda}(k)\,dk\,ds\,dt\right)\right]
\end{equation}
for the Nelson model \cite{N2}, a model for $N$ nucleons interacting with a scalar meson field, with Hamiltonian
\begin{equation}
\sum_{n = 1}^N\frac{p_n^2}{2} + \int_{k \in {\mathbb R}^3}\omega(k)a_k^{\dagger}a_k\,dk + 
\sum_{n = 1}^N\sqrt{\alpha}\int_{k \in {\mathbb R}^3}\frac{\chi_{\Lambda}}{\sqrt{\omega(k)}}(e^{ikx_n}a_k + e^{-ikx_n}a_k^{\dagger})\,dk.
\end{equation}
Here $\omega(k) = \sqrt{k^2 + \mu^2}$, where $\mu \geq 0$ is the meson mass. An ultraviolet cutoff $\chi_{\Lambda} = 1_{|k| \leq \Lambda}$ 
has been imposed in the particle-field interaction term, for otherwise the Hamiltonian would not be bounded below. The
functional integral \eqref{Nelson.eq1} is analogous  to \eqref{RF.eq1} for the polaron.

In E. Nelson's initial, remarkable work \cite{N1} on his models, the effective action was expanded into a sum of
terms: a deterministic integral, logarithmically divergent as the ultraviolet cutoff is removed, and additional
integrals, including single and double stochastic integrals, which are convergent as the ultraviolet cutoff is 
removed. He then provided bounds on the moment-generating functions for these additional terms. This
work was superseded by his later work \cite{N2} on the models, in which he employed a Gross transformation to
carry out the renormalization and show the existence of a regularized Hamiltonian by then perhaps more
transparent and familiar operator-theoretic methods.

Recently, Gubinelli, Hiroshima and L\"{o}rinczi \cite{GHL} have revisited
Nelson's stochastic integral approach for massive Nelson models or models with an additional infrared cutoff via stochastic integrals, giving a complete construction of
a regularized Hamiltonian, as well as greatly clarifying his original approach.  For small coupling they also uncover an effective attractive Yukawa-like potential between particles.  They write the effective action
for the models in terms of stochastic integrals much in the spirit of Nelson's original work, but using now well-developed methods  for estimating the moment-generating function for these integrals, e.g., the Girsanov formula, (closely related to the martingale argument we employ). 
In the analysis of \eqref{Nelson.eq1}, one encounters further expectations of the sort \eqref{general_action} given
in the introductory paragraph
with $1 < \theta < 2$. Our Theorem \ref{Theorem2} provides a bound on such expectations, log-linear in time $T$, for
large $T$.  The connection of Theorem \ref{Theorem2} with the work of Gubinelli et al., appears in section 3.4.

To give an overview, section 2 presents three theorems bounding the moment-generating functions for simple effective actions of 
particles interacting with their past trajectories and with those of other particles. Section 3 gives applications of the bounds 
to the models of quantum mechanics and non-relativistic quantum field theory previously mentioned. 
 Section 4.1 presents the main ideas behind the proofs of the bounds. Section 4.2 completes the proofs of the theorems, 
and 4.3 outlines alternative approaches to estimating functional integrals. An appendix provides  
introductory but sufficient material on Malliavin calculus and the Clark-Ocone formula for our needs,   ancillary
lemmas used in the proof of the theorems, and  the harmonic
oscillator application.

We have taken care to estimate the constants  in our theorems, particularly their 
dependence on the coupling function $f(t)$, the exponent $\theta$, and the spatial dimension $d$, 
in order to test  the theorems  numerically  in specific applications.  The method certainly accommodates other
forms for the actions, e.g., Yukawa- or Cauchy-like potentials.  
Finally, we remark that the bounds we obtain are of special utility in understanding both the 
effective attraction between particles and the question of stability of non-relativistic quantum field-theoretic Hamiltonians as a function 
of the number of particles. These applications will be addressed elsewhere.
\end{section}



\begin{section}{Theorems}
In this section, because of their similarity, our main results are stated together. We emphasize again that the bounds
are log-linear in time $T$ for large $T$. In all of the following, $T$ will be a fixed, finite, non-negative time; $f$ will denote a measurable, non-negative function 
defined on $[0, T]$; $\theta$ will be a number in $(0, 2)$; $X$ and $Y$ will denote independent, $d$-dimensional Brownian 
motions; and $d$ will be an integer greater than 1. For a function $g$, the norm $\|g\|_{p, s}$, for $1 \leq p \leq \infty$, will represent 
the $L^p$ norm of the function $g 1_{[0, s]}$.

We will make use below of the following definition. For a measurable, non-negative function $f$ on $[0, T]$, we define the 
$_*$-operation as $f_*(t) = \esssup_{t \leq s \leq T}f(s)$ -- a non-increasing version of $f$, with domain $[0, T]$. An application of $_*$ is the most natural 
way of turning an arbitrary non-negative function into a non-increasing one if one notices that $f$ is (essentially) non-increasing if and only if 
$f = f_*$ (a.e.). Notice also that this mapping stabilizes after a single application: $f_{**} = f_*$, and that it behaves well 
under scaling: if $\alpha > 0$, $(\alpha f)_* = \alpha f_*$. The map will flatten out the bumps of the 
function, from right to left. Some of our theorems will make use of 
the $_*$ operation, but the reader should always bear in mind that in most applications the function appearing will itself be non-increasing, and 
so $f_*$ will be equal to $f$ in these cases.

We will now state the main theorems. Theorem \ref{Theorem1} concerns explicit estimates for the functional integrals for sub- and 
super-Coulomb potentials with generally time-dependent coupling. Theorem \ref{Theorem2} is the analogue of Theorem \ref{Theorem1} 
with double integration. Theorem \ref{Theorem3} corresponds to Theorem \ref{Theorem2} for the case of independent Brownian motions. The theorems 
will make use of certain $\theta$- and $d$-dependent coefficients, which we now define:
\begin{eqnarray}
A_\theta\!\!\!\! & = & \!\!\!\!\frac{2^{(3\theta - 2)/(2 - \theta)}\theta^{\theta/(2 - \theta)}(2 - \theta)}{(d - \theta)^{2\theta/(2 - \theta)}},\\
B_\theta\!\!\!\! & = & \!\!\!\!\frac{\theta\Gamma[(d - \theta)/2]}{2^{\theta/2}\Gamma(d/2)},\\
C_\theta\!\!\!\! & = & \!\!\!\!\frac{2^{(3\theta - 2)/(2 - \theta)}\theta^{\theta/(2 - \theta)}(2 - \theta)}{(d - 1)^{2\theta/(2 - \theta)}},\\
D_\theta\!\!\!\! & = & \!\!\!\!\frac{\theta^{1/(2 - \theta)}\Gamma[(d - 1)/2)](d - 1)^{(2 - 2\theta)/(2 - \theta)}}
{2^{(6 - 5\theta)/(4 - 2\theta)}\Gamma(d/2)}.
\end{eqnarray}
\begin{theorem}\label{Theorem1}
Let $f$ be a non-negative measurable function. If $1 \leq \theta < 2$,
\begin{equation}
\sup_{x \in \mathbb{R}^d}E^x\left[\exp\left(\int_0^T\frac{f(t)}{|X_t|^{\theta}}\,dt\right)\right]
\label{functional_integral_1}
\end{equation}
is bounded above by
\begin{gather}
\exp\left(A_{\theta}\|f_*^{2/(2 - \theta)}\|_{1, T} + B_\theta\|f_*(t)/t^{\theta/2}\|_{1, T}\right),
\label{bound_theorem_1}
\end{gather}
whereas if $0 < \theta \leq 1$, it is bounded above by
\begin{gather}
\exp\left(
C_\theta\|f_*\|_{1, T}^{(2 - 2\theta)/(2 - \theta)}\|f_*^2\|_{1, T}^{\theta/(2 - \theta)} + 
D_\theta\left(\|f_*\|_{1, T}\|f_*^2\|_{1, T}^{-1}\right)^{(1 - \theta)/(2 - \theta)}\|f_*(t)/t^{1/2}\|_{1, T}\right).
\label{second_assertion_theorem_1}
\end{gather}
(If $\|f_*^2\|_{1, T} = 0$, then $f_* = f = 0$ almost everywhere, and therefore the functional integral \eqref{functional_integral_1} is 1.) 
Moreover, for every $0 < \theta < 2$, the supremum above is attained at $x = 0$.
\end{theorem}
\begin{theorem}\label{Theorem2}
For a non-negative measurable function $f$,
\begin{equation}
E\left[\exp\left(\int_0^T\!\!\!\int_0^t\frac{f(t - s)}{|X_t - X_s|^{\theta}}\,ds\,dt\right)\right]
\label{functional_integral_2}
\end{equation}
is bounded above by 
\begin{eqnarray}
\exp\left(A_\theta\int_0^T\|f_*\|_{1, t}^{2/(2 - \theta)}\,dt + B_\theta\int_0^T\|f_*(s)/s^{\theta/2}\|_{1, t}\,dt\right)
\end{eqnarray}
for $1 \leq \theta < 2$, and by
\begin{multline}
\exp\left\{C_\theta\left(\int_0^T\|f_*\|_{1, t}\,dt\right)^{(2 - 2\theta)/(2 - \theta)}\left(\int_0^T\|f_*\|_{1, t}^2\,dt\right)^{\theta/(2 - \theta)}
\right.\\
\hfill + \left. D_\theta\left[\left(\int_0^T\|f_*\|_{1, t}\,dt\right)\left(\int_0^T\|f_*\|_{1, t}^2\,dt\right)^{-1}\right]^{(1 - \theta)/(2 - \theta)}
\int_0^T\|f_*(s)/s^{1/2}\|_{1, t}\,dt\right\}
\label{second_assertion_theorem_2}
\end{multline}
for $0 < \theta \leq 1$. (If $\int_0^T\|f_*\|_{1, t}^2\,dt = 0$, then $f_* = f = 0$ almost everywhere, in which case 
\eqref{functional_integral_2} is 1.)
\end{theorem}
\begin{theorem}\label{Theorem3}
If $f$ is a non-negative, measurable function,
\begin{equation}
\sup_{x, y \in \mathbb{R}^d}E^{x, y}\left[\exp\left(\int_0^T\!\!\!\int_0^t\frac{f(t - s)}{|X_t - Y_s|^{\theta}}\,ds\,dt\right)\right]
\end{equation}
is bounded above by
\begin{equation}
\exp\left(2^{-\theta/(2 - \theta)}A_\theta\|f\|_{1, T}^{2/(2 - \theta)}T + 2^{-\theta/2}(1 - \theta/2)^{-1}B_\theta\|f\|_{1, T}T^{1 - \theta/2}\right)
\end{equation}
when $1 \leq \theta < 2$, and by
\begin{gather}
\exp\left(
2^{-\theta/(2 - \theta)}C_\theta\|f\|_{1, T}^{2/(2 - \theta)}T + 
2^{(4 - 3\theta)/2(2 - \theta)}D_\theta\|f\|_{1, T}^{1/(2 - \theta)}T^{1/2}
\right)
\end{gather}
when $0 < \theta \leq 1$.
\end{theorem}
\end{section}



\begin{section}{Applications}
\begin{subsection}{The hydrogen atom}
Consider the hydrogen atom Hamiltonian
\begin{gather}
H = -\frac{\Delta}{2} - \frac{\alpha}{|x|}
\end{gather}
acting in $L^2({\mathbb R}^3)$, whose ground state energy is equal to
\begin{gather}
-\lim_{T \to \infty}\frac{1}{T}\ln\left\{E\left[\exp\left(\int_0^T\frac{\alpha}{|X_t|}\,dt\right)\right]\right\}.
\end{gather}
It is well known that the value above is $-\alpha^2/2$. A direct application of Theorem \ref{Theorem1} with $f$ the constant 
$\alpha$ and $\theta$ equal to 1 yields the following bound, which is asymptotically sharp for large $T$:
\begin{gather}\label{hydrogen.eq22}
E\left[\exp\left(\int_0^T\frac{\alpha}{|X_t|}\,dt\right)\right] \leq \exp\left[\frac{\alpha^2 T}{2} + \frac{2\sqrt{2}\alpha}{\sqrt{\pi}}T^{1/2}\right].
\end{gather}
It should be stressed that this bound is for all $T$, and not just for large times. Note also that the upper bound captures both the correct long- and short-time behaviors: we have by Jensen's inequality
that
\begin{gather}
E\left[\exp\left(\int_0^T\frac{\alpha}{|X_t|}\,dt\right)\right] \geq \exp\left[E\left(\int_0^T\frac{\alpha}{|X_t|}\,dt\right)\right] = 
\exp\left(\frac{2\sqrt{2}\alpha}{\sqrt{\pi}}T^{1/2}\right).
\end{gather}
After one takes logarithms, this agrees with the right side of \eqref{hydrogen.eq22} to leading order in time $T$ for small $T$.  
\end{subsection}
\begin{subsection}{The singular potential $1/|x|^2$}
Consider the Hamiltonian
\begin{gather}
H = -\frac{\Delta}{2} - \frac{\alpha}{|x|^2}
\end{gather}
acting in $L^2({\mathbb R}^d)$, $d\geq 3$. This Hamiltonian is discussed in an elementary way in \cite{Case, EG}; see also \cite[Sections X.1 and X.2]{Simon2} for a more 
thorough treatment. If $\alpha$ is sufficiently small, one can show by simple arguments (consisting of a completion of the square) that 
$H$ can be rewritten as $A^{\dagger}A$ for some operator $A$, and therefore that $H$ is positive for small $\alpha$. Moreover, since 
$H$ has the scaling property that the replacement $x \mapsto \lambda x$ leads to $\frac{1}{\lambda^2}H$, which is unitarily equivalent 
to $H$, 
the ground state energy ($\inf {\rm spec} (H)$) is zero for small $\alpha$. (This scaling property is unique to this particular potential.) Furthermore, for sufficiently large $\alpha$ one can construct a state with negative expectation under this Hamiltonian, and the scaling property then implies
that $H$ cannot be  bounded below. There is a critical threshold for $\alpha$, $\alpha_c= (d - 2)^2/8$, dividing these two cases.
Our bounds are able to detect this sharp constant, at least in the following sense:  Let $\mathcal{E}_{\theta}^{\alpha}$ be the ground state energy 
of the Hamiltonian
\begin{gather}
-\frac{\Delta}{2} - \frac{\alpha}{|x|^{\theta}}
\end{gather}
for $1 \leq \theta < 2$. By Theorem \ref{Theorem1}, we obtain,
\begin{gather}
\mathcal{E}_{\alpha}^{\theta} \geq - 2^{2(\theta - 1)/(2 - \theta)}(2 - \theta)\theta^{\theta/(2 - \theta)}2^{\theta/(2 - \theta)}
(d - \theta)^{-2\theta/(2 - \theta)}\alpha^{2/(2 - \theta)},
\end{gather}
and therefore (using $\mathcal{E}_{\alpha}^{\theta} \leq 0$)
 \begin{equation}
         \lim_{\theta\rightarrow 2^-} \mathcal{E}_{\alpha}^{\theta} = 0
         \end{equation}
for $\alpha< \alpha_c$.
\end{subsection}
\begin{subsection}{The polaron model}
\label{polaron}
As another simple application of the bounds, consider the 1-electron polaron model, whose ground state energy is inferred from
the large-time behavior of Feynman's functional integral \cite{RF}
\begin{gather}\label{Feynmanint}
E\left[\exp\left(\frac{\alpha}{\sqrt{2}}\int_0^T\!\!\!\int_0^t\frac{e^{-(t - s)}}{|X_t - X_s|}\,ds\,dt\right)\right],
\end{gather}
which, by Theorem \ref{Theorem2}, is bounded by $e^{(\alpha + \alpha^2/4)T}$. This implies that the ground state energy is bounded below by $-\alpha - \alpha^2/4$. We emphasize that the expectation 
is with respect to {\it normalized} Brownian motion and corresponds to the Fr\"{o}hlich Hamiltonian (\ref{Frohham}) having kinetic energy operator 
$-\Delta/2$. This bound is comparable to, and a slight improvement of, Lieb and Yamazaki's bound \cite{LY}, which says that the ground state energy is bounded 
below by $-\alpha - \alpha^2/3$. It is also to be compared with the exact expression for large $\alpha$, $-0.109\alpha^2+o(\alpha^2)$ \cite{DV, LT}, the first term of which 
being the ground state energy of the Pekar functional \cite{M, L}. Also, by Jensen's inequality, and as noted 
by Feynman \cite{RF},
\begin{gather}
E\left[\exp\left(\frac{\alpha}{\sqrt{2}}\int_0^T\!\!\!\int_0^t\frac{e^{-(t - s)}}{|X_t - X_s|}\,ds\,dt\right)\right]
\geq \exp\left(\frac{\alpha}{\sqrt{\pi}}\int_0^T\!\!\!\int_0^t\frac{e^{-(t - s)}}{(t - s)}\,ds\,dt\right)\sim \exp\alpha T\end{gather}
for large $T$, and so to leading order in $\alpha$, $\alpha$ small, the ground state energy is bounded above by $-\alpha$.

The bounds can be applied to estimate the ground state energy of the bipolaron without inter-electronic repulsion (two electrons 
in a crystal lattice; see, for instance \cite[Proof of Lemma 2]{FLST}, where a functional integral representation for this model is used). 
Its Hamiltonian is given by
\begin{gather}
\sum_{n = 1}^2\frac{p_n^2}{2} + \int a_k^{\dagger}a_k\,dk + 
\sum_{n = 1}^2\frac{\sqrt{\alpha}}{2^{3/4}\pi}\int\frac{1}{|k|}(e^{ikx_n}a_k + e^{-ikx_n}a_k^{\dagger})\,dk,
\end{gather}
with functional integral
\begin{gather}
E\left[\exp\left(2\frac{\alpha}{\sqrt{2}}\int_0^T\!\!\!\int_0^t\frac{e^{-(t - s)}}{|X_t - Y_s|}\,ds\,dt + 
\frac{\alpha}{\sqrt{2}}\int_0^T\!\!\!\int_0^t\frac{e^{-(t - s)}}{|X_t - X_s|}\,ds\,dt + 
\frac{\alpha}{\sqrt{2}}\int_0^T\!\!\!\int_0^t\frac{e^{-(t - s)}}{|Y_t - Y_s|}\,ds\,dt\right)\right].
\end{gather}
Using Theorems 2 and 3 and Cauchy-Schwarz, we deduce that the expectation is bounded by
\begin{eqnarray}
\lefteqn{E\left[\exp\left(\frac{4\alpha}{\sqrt{2}}\int_0^T\!\!\!\int_0^t\frac{e^{-(t - s)}}{|X_t - Y_s|}\,ds\,dt\right)\right]^{1/2}
E\left[\exp\left(\frac{2\alpha}{\sqrt{2}}\int_0^T\!\!\!\int_0^t\frac{e^{-(t - s)}}{|X_t - X_s|}\,ds\,dt\right)\right]}\nonumber\\
 &\phantom{XXXXXXXX}\leq& \exp\left(\alpha^2T + \frac{2\alpha}{\sqrt{2}\Gamma(3/2)}T^{1/2}\right)
\exp\left(\left(2\alpha + \alpha^2  \right)T\right)\nonumber\\
&\phantom{XXXXXXXX}=& \exp\left((2\alpha + 2\alpha^2)T + \frac{4\alpha}{\sqrt{2\pi}}T^{1/2}\right),
\end{eqnarray}
which provides the lower bound $-2\alpha - 2\alpha^2$ for the bipolaron ground state energy. This result may be compared
with an upper bound that is implicit in the analysis in \cite{BB1}, provided by the Pekar-Tomasevich functional: $-0.87\alpha^2$. 
For large $\alpha$, 
the disagreement of the lower bound provided here with this upper bound is the same as in the one-polaron case (a factor of about 2.5).
\end{subsection}
\begin{subsection}{The Nelson models} 
The bounds are also applicable to functional integrals arising in the Nelson models.  As noted in 
the introduction, Gubinelli, Hiroshima, and L\"{o}rinczi \cite{GHL} also write the effective action 
for the Nelson models  with ultraviolet cutoff (and with infrared cutoffs for the massless case) using stochastic integrals.  These integrals converge as the ultraviolet cutoff is removed, ultimately resulting
in  the construction of a regularized Hamiltonian.  

Regarding this work, however, we add the following remarks. Their expansion for the action includes one-particle 
self-interaction terms, slightly simplified here for expositional purposes, each of the form
\begin{equation}
                Y^T\equiv \int_0^{T}\!\!\!\int_{(t-\tau)\wedge 0}^t\!\!\nabla_x\varphi(X_t - X_s, t-s)\,ds\,dX_t,
                 \end{equation}
with $\tau>0$ fixed, $X_t$ a three-dimensional Brownian motion,  and with   
\begin{equation}
  \varphi(x,t)\equiv \int_{{\mathbb R}^3}\frac{ e^{-ikx -\omega(k)t}}{2\omega(k)(\omega(k)+k^2/2)} dk,
   \end{equation} 
and dispersion $\omega(k)= \sqrt{\nu^2+k^2}$.  Their Lemma 2.10 in \cite{GHL}  provides
a log-linear bound on the moment-generating functional for $Y^T$. To prove the bound, they first use Girsanov's theorem 
(closely related to the proof of our  Martingale Estimate Lemma; see the next section)
to obtain the bound 
           \begin{equation}
                      E\left[ e^{\alpha Y^T}  \right]\leq E\left[ e^{\gamma Q^T}  \right]^{1/2}
                      \end{equation}
 for a suitable constant $\gamma=\gamma(\alpha,\tau)$, and 
  \begin{equation}
            Q^T \equiv \int_0^T\!\!\!\int_s^{(s + \tau)\wedge T}\!\!|X_t - X_s|^{-\theta}\,dt\,ds,
 \end{equation}
 with $1<\theta<2$.  
 (Analogous expressions are obtained for the interaction terms between particles, but with $X_t-X_s$ replaced by
 $X_t^i -X_s^j$, $i$ and $j$ labeling the particles.)     The authors then estimate $E\left[ e^{\gamma Q^T}  \right] $
 by Jensen's inequality,  
        \begin{equation}\label{Jensenmisuse}
               E\left[ e^{\gamma Q^T}  \right] \leq \frac{1}{T}\int_0^T E\left[ \exp{\left(T\gamma \int_s^{s+\tau} |X_t-X_s|^{-\theta} dt\right)}   \right]\,ds,
\end{equation}
and then state that the right side is log-linearly bounded (see their Eq.(2.39)). But since $\int_s^{s+\tau}|X_t-X_s|^{-\theta}dt$ is  an unbounded random variable, its moment-generating
functional cannot  be log-linearly bounded in $T$; in fact, one can see by scaling and Feynman-Kac that this moment generating
functional behaves as $\exp{(c\tau T^{1/(1-\theta/2)})}$ for a suitable constant $c$, hence so does the right side of
(\ref{Jensenmisuse}).     The situation can be remedied,
however.   Application of  Theorem \ref{Theorem2} with $f(s)= \chi_{[0,\tau]}(s)$
  gives the log-linear bound         
               \begin{equation}
                        E\left[ e^{\gamma Q^T}  \right] \leq \exp{\left(c(1+\gamma^{2/(2-\theta)})T\right) },
                                                \end{equation}
for $c= c(\theta, \tau)$ a suitable constant.  In particular, this argument confirms the authors' log-linear bound
on  the moment-generating functional for $Y^T$ asserted in their Lemma (2.10).

In another application to the Nelson model, a judicious manipulation of the functional integral
\begin{equation}
E\left[\exp\left(\alpha\sum_{m, n = 1}^N\int_0^T\!\!\!\int_0^t\!\!\int_{k \in \mathbb{R}^3}\frac{1}{\omega(k)}e^{ik(X_t^m - X_s^n)}
e^{-\omega(k)(t - s)}\chi_{\Lambda}(k)\,dk\,ds\,dt\right)\right]
\end{equation}
gives a lower bound for the renormalized model, as a function of the number of particles $N$ and the coupling constant $\alpha$. 
More details about this will be provided in a future article.
\end{subsection}
\end{section}



\begin{section}{Proof of the theorems}

In the following, integrals $\int\rho_t dX_t$, with $X_t$  $d$-dimensional Brownian motion, are to be regarded as It\^{o} integrals.  Our analysis will make use of the Malliavin derivative $D_u $.  A brief primer of the derivative acting on elementary functionals and
its extension to a larger class of functionals is included in the appendix. Some necessary real analysis issues are addressed there.  
 We  recall here  that given an ${\mathcal F}_T$-measurable  functional ${\mathcal A}_T$ in $L^2(\Omega)$, there 
exists a unique $\mathbb{R}^d$-valued adapted process $\rho_t$ in $L^2(\Omega\times[0, T])$, the stochastic derivative of ${\mathcal A}_T$, such that
                                     \begin{equation}\label{Decom.eq}
{\mathcal A}_T= E[{\mathcal A}_T] +\int_0^T\rho_t\, dX_t.
\end{equation}

\begin{subsection}{Principal lemmas} 
The following lemma summarizes the key ideas underlying the functional integral bounds.
\begin{lemma}(Martingale Estimate Lemma)
Let a time $T$, $0\leq T<\infty$ be given and let 
 $\mathcal{A}_T$ be a real-valued, ${\mathcal F}_T $-measurable, $L^2$-functional, and  let $\rho=\rho_t$ be its 
stochastic derivative.  
Then, if $p > 1$, 
\begin{gather}\label{GM.eq1}
E\left[e^{\mathcal{A}_T}\right] \leq e^{E[\mathcal{A}_T]}E\left[\exp\left(\frac{p^2}{2(p - 1)}\int_0^T
\rho_t^2\,dt\right)\right]^{1 - 1/p}.
\end{gather}
Moreover, if $\rho$ is in $L^{\infty}(\Omega \times [0, T])$, then
\begin{gather}\label{GM.eq2}
E\left[ e^{\mathcal{A}_T}\right] \leq e^{E[\mathcal{A}_T]}\exp\left(\frac{\|\mathcal{\rho}\|_{\infty, T}^2 T}{2}\right).
\end{gather}
\label{clark}
\end{lemma}
\begin{proof}
Since ${\cal A}_T$ is  ${\mathcal F}_T$-measurable and in $L^2$, as noted above, it can be decomposed as ${\mathcal A}_T= E[{\mathcal A}_T] +\int_0^T \rho_t \,dX_t$, 
so that by H\"{o}lder's inequality, with $p > 1$,
\begin{eqnarray}\label{lemma1.eq}
\lefteqn{E[e^{\mathcal{A}_T}] = 
E\left[\exp\left(E[\mathcal{A}_T] + \int_0^T {\rho}_t\,dX_t\right)\right]}\nonumber\\
&=& e^{E[\mathcal{A}_T]}E\left[\exp\left(\int_0^T\!\!\!\rho_t\,dX_t - \frac{p}{2}\int_0^T\rho_t^2\,dt\right)
\exp\left(\frac{p}{2}\int_0^T\rho_t^2\,dt\right)\right]\nonumber\\
&\leq& e^{E[\mathcal{A}_T]}E\left[\exp\left(\int_0^T\!\!\!p \rho_t\,dX_t - 
\frac{1}{2}\int_0^T (p\rho_t)^2\,dt\right)\right]^{1/p}
 E\left[\exp\left(\frac{p^2}{2(p - 1)}\int_0^T\rho_t^2\,dt\right)\right]^{1 - 1/p}\!\!\!\!.
\end{eqnarray}
The second factor in the last line is such that
\begin{equation}\label{supermartingale.eq}
E\left[\exp\left(\int_0^{T}\!p \rho_t\,dX_t - 
\frac{1}{2}\int_0^{T} (p\rho_t)^2\,dt\right)\right]\leq 1,
\end{equation}
 giving  the first assertion of the lemma, Ineq.\eqref{GM.eq1}.  To see that this latter inequality holds,
one can use a simple stopping-time argument:  Let $\tau^{(n)}= \tau^{(n)}(\omega)= \inf_{\{t:0\leq t \leq T\}}\{t: \int_0^t\rho_s^2(\omega)\, ds\geq n\}$ or if $ \int_0^T\rho_s^2(\omega)\, ds<n$ let $\tau^{(n)}(\omega)= T$.   Then $\tau^{(n)}\rightarrow T$, a.s., $n\rightarrow\infty$, since $\int_0^T\rho_s^2(\omega)\,ds$ is a.s. finite.
Let $\rho^{(n)}=\rho_s^{(n)}$ be $\rho_s$ up to time $\tau_n$, and  equal to zero thereafter.  Then $M_t^{(n)}= \exp\left(\int_0^t\!p \rho_s^{(n)}\,dX_s - 
\frac{1}{2}\int_0^t (p\rho_s^{(n)})^2\,ds\right)$ is a martingale for each $n$ (by, e.g., Novikov's criterion) and as a consequence has 
expectation 1 for all $t\leq T$. Inequality \eqref{supermartingale.eq} follows from an application of Fatou's lemma; the left side
is bounded by $\lim_{n\rightarrow\infty} E[M_T^{(n)}]= 1$.

The second assertion, inequality \eqref{GM.eq2}, follows by bounding $\rho$ with its $L^{\infty}$-norm and then taking the limit $p \to 1^+$.
\end{proof}
\begin{remark}
Both Nelson himself \cite{N1} and particularly Gubinelli et al. \cite{GHL} (see the proof of their Lemma 2.10) use this 
inequality as well in their analysis of the Nelson model, in the case $p = 2$.
\end{remark}
Also essential to the proofs is the following lemma:  Let $p_t(z)$ denote the $d$-dimensional heat kernel $(2\pi t)^{-d/2}\exp(-z^2/2t)$.

\begin{lemma}{(Convolution Lemma)}\label{Convolutionlemma}
Let $h$ be a bounded measurable function $[0, \infty) \to \mathbb{C}$. Then, for any $0 < \theta < 2$,
\begin{equation}
\int_0^{\infty}  h(t)\int_{\mathbb{R}^d}\frac{p_t(y)(x - y)}{|x - y|^{\theta + 2}}\,dy\,dt = a(\theta, |x|, h)\frac{x}{|x|^\theta},
\label{convolution}
\end{equation}
where
\begin{equation}
|a(\theta, |x|, h)| \leq \frac{2\|h\|_{\infty}}{\theta(d - \theta)}.
\end{equation}
\end{lemma}
\begin{proof}
 It can be directly verified that
\begin{equation}
\frac{1}{|x|^{\theta}} = \frac{(2\pi)^{d/2}}{2^{\theta/2}\Gamma(\theta/2)}\int_0^{\infty}s^{(d - \theta - 2)/2}p_s(x)\,ds,
\end{equation}
and so,
\begin{eqnarray}
\lefteqn{\int_0^\infty h(t)\int_{\mathbb{R}^d}\frac{p_t(y)(x - y)}{|x - y|^{\theta + 2}}\,dy\,dt = 
-\frac{1}{\theta}\int_0^{\infty} h(t)\nabla_x\int_{\mathbb{R}^d}\frac{p_t(y)}{|x - y|^{\theta}}\,dy\,dt}\nonumber\\
&= &-\frac{(2\pi)^{d/2}}{2^{\theta/2}\Gamma(\theta/2)\theta}\int_0^{\infty} h(t)
\nabla_x\int_{\mathbb{R}^d}\int_0^\infty p_t(y)p_s(x - y)s^{(d- \theta - 2)/2}\,ds\,dy\,dt\nonumber\\
&=& -\frac{(2\pi)^{d/2}}{2^{\theta/2}\Gamma(\theta/2)\theta}\int_0^{\infty} h(t)\nabla_x\int_0^\infty s^{(d - \theta - 2)/2}
p_{t + s}(x)\,ds\,dt\nonumber\\
&= &\frac{(2\pi)^{d/2}x}{2^{\theta/2}\Gamma(\theta/2)\theta}\int_0^{\infty} h(t)
\int_0^\infty \frac{s^{(d - \theta - 2)/2}}{t + s}p_{t + s}(x)\,ds\,dt\nonumber\\
&=& \frac{(2\pi)^{d/2}}{2^{\theta/2}\Gamma(\theta/2)\theta}\frac{x}{|x|^{\theta}}
\int_0^{\infty}\int_0^\infty h(t|x|^2)\frac{s^{(d - \theta - 2)/2}}{t + s}p_{t + s}(1)\,ds\,dt
\equiv a(\theta, |x|, h)\frac{x}{|x|^{\theta}},
\end{eqnarray}
where the substitutions $t \mapsto |x|^2 t$ and $s \mapsto |x|^2 s$ were made in the last line. Now, $a(\theta, |x|, h)$ may 
be bounded in absolute value from above as
\begin{gather}
\frac{\|h\|_{\infty}}{2^{\theta/2}\Gamma(\theta/2)\theta}\int_0^\infty\int_0^\infty
\frac{s^{(d - \theta - 2)/2}e^{-1/(2(t + s))}}{(t + s)^{d/2 + 1}}\,dt\,ds,
\end{gather}
and, with the substitution $u = 1/(2(t + s))$ for $t$, this becomes
\begin{eqnarray}
\lefteqn{\frac{2^{d/2}\|h\|_{\infty}}{2^{\theta/2}\Gamma(\theta/2)\theta}\int_0^\infty\int_0^{1/(2s)}
s^{(d - \theta - 2)/2}e^{-u}u^{d/2 - 1}\,du\,ds}\nonumber\\
&= &\frac{2^{d/2}\|h\|_{\infty}}{2^{\theta/2}\Gamma(\theta/2)\theta}\int_0^\infty\int_0^{1/(2u)}s^{(d - \theta - 2)/2}
e^{-u}u^{d/2 - 1}\,ds\,du\nonumber\\
&=& \frac{2^{d/2}\|h\|_{\infty}}{2^{\theta/2}\Gamma(\theta/2)\theta}\int_0^\infty\frac{2e^{-u}u^{d/2 - 1}}{(d - \theta)(2u)^{(d- \theta)/2}}\,du 
= \frac{2\|h\|_{\infty}}{\theta(d - \theta)},
\end{eqnarray}
which proves the assertion.
\end{proof}
\end{subsection}

\begin{subsection}{Proofs of the theorems}

\begin{subsubsection}{Proof of Theorem \ref{Theorem1}}
Let $\mathcal{A}_T = \mathcal{A}_T(f) = \int_0^T f(t)/|X_t + x|^{\theta}\,dt$. Since $f \leq f_*$ almost everywhere, 
$E[\exp(A_T(f))] \leq E[\exp(A_T(f_*))]$; we can thus assume that $f$ is non-increasing. We may also assume that $f$ is 
bounded. This last condition will be relaxed at the end of the proof. With these assumptions, 
$f = f_*$ and $\|f\|_{\infty, T} < \infty$.

Note that $\mathcal{A}_T$ is an $L^2(\Omega)$-functional; the multiple integral $E[\mathcal{A}_T^2]$, which has a non-negative integrand, can be estimated 
using the Markov property of Brownian motion and the maximality lemma \ref{Maxlemma} of the appendix.  
The computations are similar to those in the proof of the Convolution Lemma and so are omitted.
The main idea of the proof then is to use the Martingale Estimate Lemma, first computing the
expectation of
$\mathcal{A}_T$  and then estimating its conditioned Malliavin derivative. 
 
Let $^*$ denote the symmetric decreasing rearrangement operation, applied to a function \cite{LL}. (This is not to be confused with $_*$, although 
they are somewhat related, since in one dimension the outcomes can look similar in some cases.) 
Then, since for positive functions $u, v$ $\int uv \leq \int u^* v^*$, and by Fubini,
\begin{eqnarray}
\lefteqn{E\left[\int_0^T\frac{ f(t)}{|X_t+x|^\theta}\,dt\right] =\int_0^T f(t)E\left[\frac{1}{|X_t + x|^\theta}\right]\,dt = 
\int_0^T f(t)\int_{{\mathbb R}^d}\frac{1}{|y + x|^{\theta}}p_t(y)\,dy\,dt}\nonumber\\
&\leq&
\int_0^T f(t)\int_{{\mathbb R}^d}\left(\frac{1}{|y + x|^{\theta}}\right)^*p_t^*(y)\,dy\,dt
= \int_0^T f(t)\int_{{\mathbb R}^d}\left(\frac{1}{|y|^{\theta}}\right)p_t(y)\,dy\,dt \nonumber\\
& = &\displaystyle\frac{\Gamma\left(\frac{d - \theta}{2}\right)}{2^{\theta/2}\Gamma(d/2)}\int_0^T\frac{f(t)}{t^{\theta/2}}\,dt.
\end{eqnarray}

In the remainder, we will first assume that $1 \leq \theta < 2$. The stochastic derivative of $\mathcal{A}_T$, $\rho$, 
is computed as its conditioned Malliavin derivative, here a $d$-dimensional vector, which may be estimated as
\begin{eqnarray}\label{Dbound.eq}
\lefteqn{\left|E\left[D_u{\mathcal A}_T |\mathcal{F}_u\right]\right|=\left|E\left[D_u\int_0^T\frac{f(t)}{|X_t + x|^{\theta}}\,dt\Bigg|{\mathcal F}_u\right]\right|}\nonumber\\
&= &
\left|\theta\int_u^T f(t)\int_{\mathbb{R}^d}\frac{(X_u + x) - y}{|(X_u + x) - y|^{\theta + 2}}p_{s - u}(y)\,dy\,dt\right|\nonumber\\
&= &\left|\theta\int_0^{T - u}f(t + u)\int_{\mathbb{R}^d}\frac{(X_u + x) - y}{|(X_u + x) - y|^{\theta + 2}}p_{t}(y)\,dy\,dt\right|\nonumber\\
&\leq& \frac{2\|f(\cdot + u)\|_{\infty, T - u}}{(d - \theta)|X_u + x|^{\theta - 1}} = \frac{2f(u)}{(d - \theta)|X_u + x|^{\theta - 1}}.
\end{eqnarray}
In going from the third to the fourth line, we have used the Convolution Lemma. Going from the first line to the second 
is formally clear, but this step requires justification for bringing the derivative $D_u$ through the $\int dt $- integral , and then differentiating
a singular integrand. The argument involves mollifying the integrand and approximating the integral by a Riemann integral.  
See the appendix for a sketch of  the argument. Therefore, with $C \equiv 2p^2/(p - 1)(d - \theta)^2$, we get, from the Martingale Estimate Lemma,
\begin{equation}
E[e^{\mathcal{A}_T}] \leq e^{E[\mathcal{A}_T]}E\left[\exp\left(C\int_0^T\frac{f(t)^2}
{|X_t + x|^{2\theta - 2}}\,dt\right)\right]^{1 - 1/p}.
\label{first_bound_martingale_estimate}
\end{equation}

Now, we  use Young's inequality, $ab \leq a^q/q + b^r/r$ for non-negative $a$ and $b$, and $q$ and $r$ conjugate ($q^{-1} + r^{-1} = 1$) as follows:
\begin{gather}
\frac{Cf(t)^2}{|X_t + x|^{2\theta - 2}} = \frac{Cf(t)^{1/q + 1}}{r^{1/r}}\cdot\frac{f(t)^{1/r}r^{1/r}}{|X_t + x|^{2\theta - 2}}
\leq \frac{C^q f(t)^{q + 1}}{q r^{q/r}} + \frac{f(t)}{|X_t + x|^{(2\theta - 2)r}}.
\end{gather}
Then, with $(2\theta - 2)r = \theta$ and $q$ equal to the conjugate of $r$, we obtain
\begin{gather}
E[e^{\mathcal{A}_T}]
\leq e^{E[\mathcal{A}_T]}E\left[\exp\left(D_{\theta}\int_0^T f(t)^{2/(2 - \theta)}\,dt\right)
\exp\left(\int_0^T\frac{f(t)}{|X_t + x|^{\theta}}\,dt\right)
\right]^{1 - 1/p}
\end{gather}
with $D_{\theta}= \frac{2-\theta}{\theta} \left(\frac{\theta}{2(\theta-1)}\right)^{-2(\theta-1)/(2-\theta)}C^{\theta/(2 - \theta)}$, ($D_1= C$),  i.e., 
\begin{gather}
E[e^{\mathcal{A}_T}] \leq e^{p E[\mathcal{A}_T]}\exp\left[(p - 1)D_{\theta}\|f^{2/(2 - \theta)}\|_{1,T}\right].
\end{gather}
In this last step it was assumed that $E[e^{\mathcal{A}_T}]$ is finite. It is easy to see why this is true. First note that if $\theta = 1$, 
then equation \eqref{first_bound_martingale_estimate} shows that $E[e^{\mathcal{A}_T}]$ is finite. If $1 < \theta < 2$, a finite number 
of applications of the bound \eqref{first_bound_martingale_estimate} will lower the exponent in the denominator to a number in $[0, 1]$. 
If the number lies in $(0, 1)$, an argument involving Young's inequality (see below, in this same proof) will raise the exponent to 1, 
in which case finiteness will be concluded. Setting now $p = \theta$, in order to minimize $(p - 1)D^{\theta/(2 - \theta)}$, completes 
the proof of the theorem for $\theta \geq 1$.

If $0< \theta < 1$, we compare with the $\theta=1$ case by using Young's inequality, as mentioned above, this time picking 
$q = 1/(1 - \theta)$, $r = 1/\theta$, $a = f^{1/q}/(\varepsilon^\theta r^{1/r})$, $b = f^{1/r}r^{1/r}\varepsilon^{\theta}/|X_t + x|^{\theta}$, to get
\begin{eqnarray}
\lefteqn{E\left[\exp\left(\int_0^T\frac{f(t)}{|X_t + x|^{\theta}}\,dt\right)\right]}\nonumber\\
&&\leq \exp\left(\frac{(1-\theta)\theta^{\theta/(1-\theta)}}{ \varepsilon^{\theta/(1-\theta)}} \int_0^Tf(t)\,dt\right)
E\left[\exp\left(\varepsilon\int_0^T\frac{f(t)}{|X_T+x|}\,dt\right)\right].
\end{eqnarray}
  But by the $\theta=1$-case above, this is in turn bounded by
\begin{gather}
  \exp\left(\frac{(1 - \theta)\theta^{\theta/(1 - \theta)}}{\varepsilon^{\theta/(1 - \theta)}}\|f\|_{1,T} + \frac{2\varepsilon^2}
{(d - 1)^2}\|f\|_{2, T}^2 + \frac{\Gamma((d - 1)/2)\varepsilon}{2^{1/2}\Gamma(d/2)}\int_0^T\frac{f(t)}{t^{1/2}}\,dt\right).
\label{pre_optimization_theorem_1}
\end{gather}
This concludes the proof in the case of a non-increasing, bounded function $f$. Suppose now that $f$ is again non-increasing, but this time 
unbounded. Then $f$ must be infinite on an interval $[0, S]$, and finite otherwise. If $S > 0$, then all the norms on the right hand side of the inequality we are proving 
are infinite, and therefore the assertion is true. If $S = 0$, then one may truncate $f$ by considering the family of functions 
$f_a(t) = f(t)$ if $a \leq  t \leq T$ and $f_a(t) = f(a)$ if $0 \leq t \leq a$. Then $f_a$ is bounded and non-increasing for all $0 < a \leq T$. The 
assertion then follows by applying the result for bounded and non-increasing functions and then taking the limit $a \to 0^+$. This is justified 
by monotone convergence. The final part of the theorem follows from an optimization in $\varepsilon$ of the first two terms in equation \eqref{pre_optimization_theorem_1} 
(leading to $\varepsilon = \theta^{1/(2 - \theta)}\left((d - 1)^2\|f\|_{1, T}/(4\|f\|_{2, T}^2)\right)^{(1 - \theta)/(2 - \theta)}$). 
The statement concerning maximality is proven in the appendix.
\qed
\end{subsubsection}


\begin{subsubsection}{Proof of Theorem \ref{Theorem2}}
The proof in this case is very similar to that of the previous theorem. Here the action is given by $\mathcal{A}_T=\int_0^T\!\!\int_0^t\frac{f(t - s)}{|X_t - X_s|^{\theta}}\,ds\,dt$. 
We again assume that $f$ is non-increasing and bounded. This restriction may be eliminated in the way detailed in the previous proof, and an 
explanation of how this is done in this case will be omitted. We will again consider first the case $1 \leq \theta < 2$ and then $0 < \theta < 1$. 
One can check that $\mathcal{A}_T$ is indeed in $L^2(\Omega)$. The use of the Martingale Estimate Lemma is again the main component of the proof.
We then determine $E[\mathcal{A}_T]$:
\begin{gather}
E[\mathcal{A}_T] = \int_0^T\!\!\!\int_0^t f(t - s)E\left(\frac{1}{|X_t - X_s|^{\theta}}\right)\,ds\,dt = 
\int_0^T\!\!\!\int_0^t f(t - s)\int_{{\mathbb R}^d}\frac{1}{|y|^{\theta}}p_{t - s}(y)\,dy\,ds\,dt\nonumber
\end{gather}
\begin{eqnarray}
& = & \frac{1}{(2\pi)^{d/2}}\int_0^T\!\!\!\int_0^t \frac{f(t - s)}{(t - s)^{\theta/2}}\int_{{\mathbb R}^d}\frac{e^{-y^2/2}}{|y|^{\theta}}\,dy\,ds\,dt\nonumber\\
& = & \frac{|\mathbb{S}^{d - 1}|}{(2\pi)^{d/2}}\int_0^T\!\!\!\int_0^t\frac{f(t - s)}{(t - s)^{\theta/2}}\int_0^{\infty}e^{-r^2/2}r^{d - 1 - \theta}\,dr\,ds\,dt\nonumber\\
& = & \frac{\Gamma[(d - \theta)/2]|\mathbb{S}^{d - 1}|}{2^{\theta/2 + 1}\pi^{d/2}}\int_0^T\!\!\!\int_0^t\frac{f(t - s)}{(t - s)^{\theta/2}}\,ds\,dt\nonumber\\
& = & \frac{\Gamma[(d - \theta)/2]|\mathbb{S}^{d - 1}|}{2^{\theta/2 + 1}\pi^{d/2}}\int_0^T\!\!\!\int_0^t\frac{f(s)}{s^{\theta/2}}\,ds\,dt\nonumber\\
& = &\frac{\Gamma[(d - \theta)/2]}{2^{\theta/2}\Gamma(d/2)}\int_0^T\!\!\!\int_0^t\frac{f(s)}{s^{\theta/2}}\,ds\,dt
\equiv K\int_0^T\!\!\!\int_0^t\frac{f(s)}{s^{\theta/2}}\,ds\,dt.
\end{eqnarray}
Next, we proceed to compute the stochastic derivative $\rho$ of $\mathcal{A}_T$, which, as before, will be equal to $E[D_u\mathcal{A}_T|\mathcal{F}_u]$:
\begin{eqnarray}\label{Conveq.12}
\lefteqn{E[D_u\mathcal{A}_T|\mathcal{F}_u] = 
-\theta\int_u^T\!\!\!\int_0^u f(t - s)E\left[\frac{(X_t - X_s)}{|X_t - X_s|^{\theta + 2}}\right|\mathcal{F}_u\Bigg]\,ds\,dt}\nonumber\\
&=& -\theta\int_u^T\!\!\!\int_0^u f(t - s)E\left[\frac{(X_t - X_u) + (X_u - X_s)}{|(X_t - X_u) + (X_u - X_s)|^{\theta + 2}}\right|
\mathcal{F}_u\Bigg]\,ds\,dt\nonumber\\
&=& -\theta\int_u^T\!\!\!\int_0^u f(t - s)\int_{{\mathbb R}^d}\frac{y + (X_u - X_s)}{|y + (X_u - X_s)|^{\theta + 2}}p_{t - u}(y)\,dy\,ds\,dt\nonumber\\
&=& \theta\int_0^u\!\!\!\int_0^{T - u}f(t + u - s)\int_{{\mathbb R}^d}\frac{(X_s - X_u) - y}{|(X_s - X_u) - y|^{\theta + 2}}p_t(y)\,dy\,dt\,ds\nonumber\\
&=& \theta\int_0^u\|f(\cdot + u - s)\|_{\infty, T - u}a(\theta, |X_s - X_u|, 1_{[0, T - u]})\frac{X_s - X_u}{|X_s - X_u|^{\theta}}\,ds,
\end{eqnarray}
where the convolution lemma was used in the last line. The first line requires justification: 
The derivative has been brought through the double integral and then acted upon a singular
integrand. The step is justified by considering limiting mollified actions analogous to those discussed in the appendix for the single
integral case.  Proceeding, we use the inequality
\begin{gather}
\|f(\cdot + u - s)\|_{\infty, T - u} = \esssup_{t \in [u - s, T - s]}f(t) \leq f_*(u - s) = f(u - s),
\end{gather}
and Cauchy-Schwarz, to get
\begin{gather}
\big|E[D_u\mathcal{A}_T|\mathcal{F}_u]\big| \leq \frac{2}{d - \theta}\int_0^u\frac{f(u - s)}{|X_u - X_s|^{\theta - 1}}\,ds 
\leq \frac{2}{d - \theta}\|f\|_{1, u}^{1/2}\left(\int_0^u\frac{f(u - s)}{|X_u - X_s|^{2\theta - 2}}\,ds\right)^{1/2}.
\end{gather}
The Martingale Estimate Lemma tells us then that
\begin{eqnarray}
\lefteqn{E\left[\exp\left(\int_0^T\!\!\!\int_0^t\frac{f(t - s)}{|X_t - X_s|^{\theta}}\,ds\,dt\right)\right]}\nonumber\\
&\leq& \exp\left(K\int_0^T\!\!\!\int_0^t\frac{f(s)}{s^{\theta/2}}\,ds\,dt\right)E\left[\exp\left(\frac{2p^2}{(d - \theta)^2(p - 1)}
\int_0^T\!\!\!\int_0^u\|f\|_{1, u}\frac{f(u - s)}{|X_u - X_s|^{2\theta - 2}}\,ds\,du\right)\right]^{1 - 1/p}\nonumber\\
&\equiv &\exp\left(K\int_0^T\!\!\!\int_0^t\frac{f(s)}{s^{\theta/2}}\,ds\,dt\right)E\left[\exp\left(L\int_0^T\|f\|_{1, u}\int_0^u
\frac{f(u - s)}{|X_u - X_s|^{2\theta - 2}}\,ds\,du\right)\right]^{1 - 1/p},
\end{eqnarray}
which shows in particular that $E[\exp(\mathcal{A}_T)]$ is finite, as noted in the proof of the previous theorem. 
We now proceed as before, using Young's inequality, assuming that $\theta\geq 1$:
\begin{eqnarray}
L\int_0^T\|f\|_{1, u}\int_0^u\frac{f(u - s)}{|X_u - X_s|^{2\theta - 2}}\,ds\,du&\!\!\!\leq&\!\!\! \frac{2-\theta}{\theta} \left(\frac{\theta}{2(\theta-1)}\right)^{-2(\theta-1)/(2-\theta)}L^{\theta/(2 - \theta)}\int_0^T\|f\|_{1, u}^{2/(2 - \theta)}\,du\nonumber\\
&&\phantom{XXXX} +\int_0^T\!\!\!\int_0^u\frac{f(u - s)}{|X_u - X_s|^{\theta}}\,ds\,du.
\end{eqnarray}
We obtain finally
\begin{eqnarray}
\lefteqn{E\left[e^{{\mathcal A}_T}\right] =E\left[\exp\left(\int_0^T\!\!\!\int_0^t\frac{f(u - s)}{|X_t - X_s|^{\theta}}\,ds\,dt\right)\right]}\nonumber\\
&\leq& \exp\left(pK\int_0^T\|f(s)/s^{\theta/2}\|_{1, t}\,dt + (p - 1)c_{\theta}L^{\theta/(2-\theta)}\int_0^T\|f\|_{1, t}^{2/(2-\theta)}\,dt\right)
\end{eqnarray}
with $c_{\theta}= \frac{2-\theta}{\theta} \left(\frac{\theta}{2(\theta-1)}\right)^{-2(\theta-1)/(2-\theta)}$, ($c_1= 1$). 
We pick $p= \theta$ to minimize $(p-1)L^{\theta/(2-\theta)}$, and the result follows, again for $\theta\geq 1$.

Now, assume that $0 < \theta < 1$. Again applying Young's inequality to the action and comparing with the $\theta = 1$ case, we get
\begin{eqnarray}
\lefteqn{E\left[e^{\mathcal{A}_T}\right] = E\left[\exp\left(\int_0^T\!\!\!\int_0^t\frac{f(t - s)}{|X_t - X_s|^{\theta}}\,ds\,dt\right)\right]}\nonumber\\
&\leq&\exp\left((1-\theta)\theta^{\theta/(1-\theta)}\varepsilon^{-\theta/(1 - \theta)}\int_0^T\|f\|_{1, t}\,dt\right)
E\left[\exp\left(\varepsilon\int_0^T\!\!\!\int_0^t\frac{f(t - s)}{|X_t - X_s|}\,ds\,dt\right)\right]\nonumber\\
&\leq&\exp\left\{(1-\theta)\theta^{\theta/(1-\theta)}\varepsilon^{-\theta/(1-\theta)}\int_0^T\|f\|_{1, t}\,dt
+ \frac{\Gamma((d - 1)/2)}{2^{1/2}\Gamma(d/2)}\varepsilon\int_0^T\|f(s)/s^{1/2}\|_{1, t}\,dt\right.\nonumber\\
&& \qquad\,\, + \left.\frac{2\varepsilon^2}{(d-1)^2}\int_0^T\|f\|_{1, t}^2\,dt\right\}.
\end{eqnarray}
The parameter $\varepsilon$ is now chosen to minimize the sum of the first two terms ($\varepsilon$ having the 
same numerical coefficient as the $\varepsilon$ of the previous proof). This completes the proof for $0 < \theta < 1$.
\qed
\end{subsubsection}


\begin{subsubsection}{Proof of Theorem \ref{Theorem3}}

We proceed as in the proof of \cite[Lemma 2]{FLST}: With a simple change of variable $u = t - s$,
\begin{gather}
E\left[\exp\left(\int_0^T\!\!\!\int_0^t\frac{f(t - s)}{|X_t - Y_s + (x - y)|^{\theta}}\,ds\,dt\right)\right] = 
E\left[\exp\left(\int_0^T f(u)\int_0^{T - u}\frac{ds}{|X_{u + s} - Y_s + (x - y)|^{\theta}}\,du\right)\right]\nonumber\\
= E\left[\exp\left(\frac{1}{\|f\|_{1, T}}\int_0^T f(u)\int_0^{T - u}\frac{\|f\|_{1, T}}{|X_{u + s} - Y_s + (x - y)|^{\theta}}\,ds\,du\right)\right]\nonumber\\
\leq \frac{1}{\|f\|_{1, T}}\int_0^T f(u)E\left[\exp\left(\int_0^T\frac{\|f\|_{1, T}}{|X_{u + s} - Y_s + (x - y)|^{\theta}}\,ds\right)\right]\,du,\label{markov}
\end{gather}
where Jensen's inequality is used from the second to the last line.
This last line can be written as
\begin{eqnarray}
\lefteqn{\frac{1}{\|f\|_{1, T}}\int_0^Tf(u)E\left[\exp\left(\int_0^T\frac{\|f\|_{1, T}}{|X_u + \left((X_{u+s}-X_u) - Y_s\right) + (x - y)|^{\theta}}\,ds\right)\right]\,du}\nonumber\\
&\equiv& \frac{1}{\|f\|_{1, T}}\int_0^Tf(u)E\left[\exp\left(\int_0^T\frac{\|f\|_{1, T}}{|X_u + \sqrt 2 W^{(u)}_s + (x - y)|^{\theta}}\,ds\right)\right]\,du,
\end{eqnarray}
where $W^{(u)}_s\equiv \frac{1}{\sqrt 2}(X_{u+s} - X_u - Y_s)$ is, for fixed $u$, a standard $d$-dimensional Brownian motion 
independent of $X_u$. By the maximality result in the second part of the appendix, the above is bounded simply by
\begin{equation}
E\left[\exp\left(\frac{\|f\|_{1, T}}{2^{\theta/2}}\int_0^T\frac{ds}{|W^{}_s|^{\theta}}\right)\right].
\end{equation}
An application of Theorem \ref{Theorem1} completes the proof of Theorem \ref{Theorem3}, with the role of $f$ played 
by the constant $2^{-\theta/2}\|f\|_{1, T}$ on the interval $[0,T]$.\qed
\begin{remark}
This theorem can certainly be proved in the manner of the first two theorems. To give a glimpse of what 
is involved, the expectation term in the Clark-Ocone formula can be computed by a direct application of the Hardy-Littlewood-Sobolev inequality:
   \begin{gather}
  E^{x, y}\left[\mathcal{A}_T(f,X,Y,\theta)\right] = E\left[\int_0^T\!\!\!\int_0^t\frac{f(t - s)}{|X_t - Y_s + (x - y)|^{\theta}}\,ds\,dt\right]\nonumber\\
= \int_0^T\!\!\!\int_0^t\!\!\iint\frac{f(t - s)p_t(a)p_s(b)}{|(a - b) + (x - y)|^{\theta}}\,da\,db\,ds\,dt 
= \int_0^T\!\!\!\int_0^t\!\!\iint\frac{f(t - s)p_t(a - x)p_s(b - y)}{|a - b|^{\theta}}\,da\,db\,ds\,dt\nonumber\\ 
\leq C_{HLS}(\theta)\int_0^T\!\!\!\int_0^tf(t - s) \| p_{t} \|_{p} \| p_s\|_{p }\,ds\,dt
 = C_{HLS}(\theta) \,p^{-d/p} (2\pi)^{-d/q} \int_0^T\!\!\!\int_0^t\frac{f(t - s)}{t^{d/2q} s^{d/2q}}\,ds\,dt,
\end{gather}
where $C_{HLS}(d,\theta)$ is the Hardy-Littlewood-Sobolev constant \cite{LL} with $p$ determined by 
$2/p +\theta/d=2$, $q$ is conjugate to $p$, and $p_t$ is the heat kernel.   The conditioned Malliavin derivative
of the action can be estimated  similarly.  But the resulting  estimate of the coefficient of $T$ in the exponential
is smaller (better) via the Jensen's inequality approach used here than via  the martingale approach, the former proof making better
use of the independence of the two Brownian motions. For this reason, and because of its simplicity, we present only the proof given above which does, however, use Theorem \ref{Theorem1}.
\end{remark}
\end{subsubsection}
\end{subsection}
\begin{subsection}{Remarks on alternative approaches to double integral estimates}
There are alternative approaches to estimating the moment-generating functions for the actions considered in
this article, for instance the self-attracting case
\begin{equation}
       \mathcal{A}_T = \int_0^T\!\!\!\int_0^t\frac{f(t-s)}{|X_t-X_s|^\theta}\,ds\,dt
\label{self_attracting}
              \end{equation}
(the action of Theorem \ref{Theorem2}), particularly methods drawing on the work on self-intersecting (renormalized) local times \cite{V} for planar Brownian motion, c.f.,  Le Gall \cite{LG1, LG2}, and Marcus and Rosen \cite{MR}. Le Gall's method, adapted to our
purposes,              
 begins by decomposing the upper triangular region of integration $R\equiv \{(s,t):\, 0\leq s\leq t\leq T\}$ into an infinite union
of disjoint dyadic squares
               \begin{equation}
                       R= \bigcup_{n\geq 1}\bigcup_{k=1}^{2^{n-1}} R_{n,k},
                       \end{equation}
with 
\begin{equation}
  R_{n,k} = \left\{(s,t):  \frac{(2k-2)T}{2^n}< s\leq \frac{(2k-1)T}{2^n},\,\,\,\frac{(2k-1)T}{2^n}<t\leq \frac{2kT}{2^n}\right\},
\end{equation}                     
$k= 1,2, ...,2^{n-1}$, so that
in an obvious notation,
\begin{eqnarray}
\mathcal{A}_T & = & \sum_{n\geq 1}\sum_{k = 1}^{2^{n - 1}}{\cal A}(R_{n,k})\nonumber\\
 &\equiv & \sum_{n\geq 1}\sum_{k=1}^{2^{n-1}}\iint\limits_{R_{n,k}}\frac{f(t-s)}{|X_t - X_s|^\theta}\,ds\,dt.
\end{eqnarray}
Note that for fixed $n$, the ${\cal A}(R_{n,k})$'s, $k= 1,...,2^{n-1}$ are i.i.d.  Assuming a suitable (approximate) scaling
relation for $f$, one can then relate  ${\cal A}(R_{n,k})$ to ${\cal A}(R_{1,1})$.  An application of H\"{o}lder's
inequality for an infinite product gives an upper bound on $E\left[\exp {\cal A}_T\right]$ in terms of the moment-generating function $E\left[\exp(\alpha {\cal A}(R_{1,1}))\right] $ for ${\cal A}(R_{1,1})$.  
Estimating the latter expectation, in which $s$ and $t$ essentially stay away from each other, presents simpler, although non-trivial, tasks.  For 
$f$ integrable, one can use conditioning with respect to ${\cal F}_{T/2}$, Feynman-Kac for 
Schr\"{o}dinger operators, and then operator and differential equation methods to estimate $E\left[\exp {\cal A}_T\right] $. Note that Le Gall also provides a {\it lower} bound 
for the exponential moment of $\mathcal{A}_T$. By contrast, the stochastic integral
approach provides a fairly direct route for bounding $E\left[\exp {\cal A}_T\right]$ from above, particularly
its log-linear behavior in $T$. (We would like to thank D. Brydges for calling our attention to this alternative approach.)

In the event that the function $f$ extends to a symmetric ($f(t)= f(-t)$) positive definite function, 
one may construct a simple {\it lower} bound to the exponential moment of the self-attracting action \eqref{self_attracting}, 
valid for large time $T$ and agreeing at least in the large coupling regime with the upper bound presented above (Theorem \ref{Theorem2}). Since $1/|x|^{\theta}$ is itself positive definite,
it follows from Young's inequality that
                    \begin{eqnarray}
          \lefteqn{ {\cal A}_T = \int_0^T\!\!\!\int_0^t \frac{f(t-s)}{|X_t-X_s|^{\theta}} \,ds\,dt}\nonumber\\
          &\geq &\int_{y\in {\mathbb R}^d} \int_0^T\!\!\!\int_0^T\frac{f(t-s)\xi(y)}{|X_t-y|^{\theta}} dy\,ds\,dt\nonumber\\&& -\frac{1}{2}  \int_0^T\!\!\!\int_0^Tf(t-s)\,ds\,dt \int_{y\in {\mathbb R}^d}\int_{x\in {\mathbb R}^d} \frac{\xi(x)\xi^*(y)}{|X_t-X_s|^{\theta}}  dx\,dy\nonumber\\
          &=& 2\|f\|_1\int_0^T\!\!\!\int_{y\in {\mathbb R}^d} \frac{\xi(y)}{|X_t-y|^{\theta}} dy\,dt -\|f\|_1 T\int_{y\in {\mathbb R}^d}\int_{x\in {\mathbb R}^d} \frac{\xi(x)\xi^*(y)}{|X_t-X_s|^{\theta}}  dx\,dy +o(T)             \end{eqnarray}
for any $\xi$ (and with $\|f\|_1 =\int_0^{\infty}|f(t)| dt$). Thus, by Feynman-Kac, 
\begin{eqnarray}
\lefteqn{E\left[ e^{{\cal A}_T}  \right] }\nonumber\\
& \geq& \exp\left(-\|f\|_1 T\int_{y\in {\mathbb R}^d}\int_{x\in {\mathbb R}^d} \frac{\xi(x)\xi^*(y)}{|X_t-X_s|^{\theta}}  dx\,dy +o(T)        \right) E\left[ \exp\left(-\int_0^T V_{\xi}(X_t)\,dt\right)   \right]\nonumber\\
   &\simeq& \exp{(- T{\cal E}_{f,\theta}}),
\end{eqnarray}
where $V_{\xi}(x)= 2\|f\|_1\frac{1}{|\cdot|^{\theta}}*\xi(x)$, $\xi$ is chosen optimally (ultimately as $\psi^2$, see below), 
and ${\cal E}_{f,\theta}$ is given by a Pekar functional in complete analogy with the optical polaron case (alluded to in subsection \ref{polaron}),
\begin{equation}
 {\cal E}_{f,\theta}= \inf_{\{\psi: \|\psi\|_2= 1\}}\left(\frac{1}{2}\int_{{\mathbb R}^d}
|\nabla \psi|^2 dx-\|f\|_1\int_{{\mathbb R}^d}\int_{{\mathbb R}^d}\frac{\psi(x)^2 \psi(y)^2}{|x-y|^{\theta}}\,dx\,dy
\right).\end{equation}
We note here that ${\cal E}_{f,\theta}$ scales as $\|f\|_1^{2/(2-\theta)}$, in agreement with
the scaling of the first term in the upper bound provided by Theorem 2.2. (We thank E. Lieb for this argument
as it pertains to the polaron case. See also \cite{DV}, where positive definiteness of $f$ was used to obtain a lower bound 
on the functional integral.)
\end{subsection}
\vspace{10pt}
\begin{acknowledgment}
This material is in part based upon work supported by the National Science Foundation under Grant No. 0932078000 while LET was in residence at the Mathematical Sciences Research Institute in Berkeley, California, during the Fall 2015 semester.
\end{acknowledgment}
\end{section}



\begin{appendix}
\setcounter{section}{1}
\setcounter{equation}{0}
\begin{section}*{Appendix}
\begin{subsection}{Introduction to Malliavin calculus and the Clark-Ocone formula}
The following is intended as a self-contained introduction to the elements of Malliavin calculus used throughout the article. A standard 
reference for it is \cite{N}. The framework consists of the finite time horizon Wiener space 
$(\Omega, \mathcal{F}_T, P)$, with $\Omega = C([0, T])$; $\mathcal{F}_T = \sigma(X_t : 0 \leq t \leq T)$, where $X$ is the standard $d$-dimensional 
Brownian motion on $\Omega$, namely $X_t(\omega) = \omega(t)$; and $P$ is Wiener measure. A real-valued 
function $F$ in $L^2(\Omega)$ is called a {\it Brownian functional.} For a Brownian functional $F$ which is
${\mathcal F}_T$ measurable, there is a unique $\mathbb{R}^d$-valued, adapted $L^2$-process $\rho = \rho_t$ such that
\begin{gather}
F =  E[F]+ \int_0^T \rho_t\,dX_t.
\label{clark_ocone}
\end{gather}
This, in a sense, is a ``fundamental theorem of calculus'' for Brownian functionals. We call $\rho_t$ the stochastic derivative of $F$. The 
theorem is stated and proved in standard references on stochastic analysis, such as \cite[Section 3.4]{KS1}; see also \cite{N} and \cite{KS2}.

Despite the appeal of the formula \eqref{clark_ocone}, it is not clear at the outset how one can calculate the stochastic derivative of a given function $F$ in $L^2(\Omega)$. 
There is, however, a particular class of functions for which it is possible to explicitly compute $\rho_t$, which we now describe. 
Let $H$ be the Hilbert space of square integrable $\mathbb{R}^d$-valued functions $H = L^2([0,T],dt)$. Define the mapping $H \to L^2(\Omega)$ given by
\begin{gather}
g \mapsto W(g) \equiv \int_0^T g_t\,dX_t.
\label{fundamental_theorem}
\end{gather}
This mapping is an isometry onto a closed subspace of $L^2(\Omega)$ and defines a Gaussian process on $H$. In particular,
\begin{gather}
E\left[W(g)W(h)\right] = (g, h)_H.
\end{gather}
Armed with these notions, we now define $\mathcal{S}$ as the space of functions in $L^2(\Omega)$ of the form
\begin{gather}
F = f\left(W(g_1), \ldots, W(g_m)\right),
\label{L_2_function_1}
\end{gather}
for some smooth function $f: \mathbb{R}^m \to \mathbb{R}$ with polynomial growth and functions $g_1, \ldots, g_m$ in $H$. We
define the Malliavin derivative operator $D$ acting on functions in  $\mathcal{S}$ and taking values in $L^2(\Omega \times [0, T])$ as follows: For $F$ of the form \eqref{L_2_function_1},  $DF$ evaluated at a given time $t$
is given by
\begin{gather}
D_tF = \sum_{i = 1}^n\frac{\partial f}{\partial x_i}\left(W(g_1), \ldots, W(g_m)\right)g_i(t).
\label{formula_derivative}
\end{gather}
 It can be shown that $D$ is a closable (unbounded) operator 
with domain $\mathcal{S}$. Finally, let us define $\mathbb{D}$ as the subpace of $L^2(\Omega)$ given by the square integrable 
functions $F$ such that there is a sequence $F_j$ in $\mathcal{S}$ that converges to $F$ in $L^2$, and also with the property that $DF_j$ is Cauchy 
(and therefore converges to an element in $L^2(\Omega \times [0, T])$). $\mathbb{D}$ amounts to the completion of $\mathcal{S}$ under the inner product
\begin{equation}
(F, G) = E[FG] + E[(DF, DG)_H].
\end{equation}
The fact that $D$ is closable as an operator from $\mathcal{S}$ to $L^2(\Omega)$ implies that $D$ can be extended as an unbounded operator to all of 
$\mathbb{D}$ as follows: $DF = \lim_{j \to \infty}DF_j$, where $F_j$ is any sequence such that $F_j \to F$ in $L^2$ and $DF_j$ is $L^2(\Omega \times [0,T])$-Cauchy.

We are now in a position to state the formula that concerns us in this paper: For any $F$ in $\mathbb{D}$,
\begin{gather}\label{Bigeqn}
F = E(F) + \int_0^T E(D_t F|\mathcal{F}_t)\,dX_t,
\end{gather}
which provides a means for computing its stochastic derivative. This identity is called the {\it Clark-Ocone formula.}

The definition for a stochastic derivative given above, \eqref{formula_derivative}, is not of direct use
in computing the stochastic derivatives of the actions encountered in this work, 
since they are typically not in $\mathcal{S}$. In all cases, however, 
the stochastic derivatives can be rigorously calculated. We sketch now, as an example, the computation of the stochastic derivative 
of a typical action we encounter in the text,
\begin{equation}
\mathcal{A}_T = \int_0^T\frac{dt}{|X_t|^{\theta}},
\end{equation}
which is in $L^2(\Omega)$ for $\theta<2$. Noting the obvious $ D^i_t X_s = 1$ for $t<s$,  $=0$ otherwise,  $i= 1,2,..,d$, the Malliavin derivative of ${\cal A}_T$ is given formally by 
     \begin{equation}\label{appendixeq.92}
         D_t{\mathcal A}_T= -\theta \int_t^T\frac{X_s}{|X_s|^{\theta +2}}ds, 
     \end{equation}
and is in $L^2(\Omega\times [0,T])$ for $\theta < d/2$, which may be seen by means of the Convolution Lemma. For $\theta$ in this range, 
the derivative above is not merely formal but actually the right answer. This is 
justified if one considers the limit  $\varepsilon\rightarrow 0$ of mollified actions ${\cal A}_T(\varepsilon) = 
\int_0^T{(X_t^2+\varepsilon^2)^{-\theta/2}}\,dt$; the derivative comes through the integration in the mollified 
integrals by approximating them as Riemann sums. This implies, in particular, that $\mathcal{A}_T$ is in $\mathbb{D}$ when $\theta < d/2$. 
However, if $\theta \geq d/2$, $D_t{\mathcal A}_T$, as computed formally in 
\eqref{appendixeq.92}, is not in $L^2(\Omega\times [0,T])$ -- this can be verified again using the Convolution Lemma. 
Nevertheless, the conditioned  derivative $E[D_t{\mathcal A}_T|{\cal F}_t ]$, and which
with some abuse of notation we continue to write in this manner, is again obtained as an $L^2$-limit of mollified derivatives, and 
is given by 
             \begin{equation}
               -\theta \int_{y\in {\mathbb R}^d}\int_t^T\frac{y+X_t}{|y+X_t|^{\theta +2}} p_{s-t}(y)\,ds\,dy.   
              \end{equation}
This expression is bounded by $const \times |X_t|^{-\theta+1}$, again from the Convolution Lemma, and is indeed in $L^2(\Omega \times [0, T])$
for $\theta<  2$.  Eq.(\ref{Bigeqn}) holds as well for ${\mathcal A}_T$ via a straightforward
dominated convergence argument  using  mollified actions. All other derivative computations in the text may also be justified in this way.
\end{subsection}

\begin{subsection}{Maximality of functional integrals at $x = 0$}

This section of the appendix provides a proof of  the maximality of
\begin{gather}\label{appendix1.eq}
E^x\left[\exp\left(\int_0^T \frac{f(t)}{|X_t|^{\theta}}\,dt\right)\right]
\end{gather}
at $x = 0$, when $0 \leq \theta < 2$, $f$ non-negative.  We will actually prove this for more general functionals. See Glimm and Jaffe \cite{GJ} for basic formulae for Gaussian integrals.  
\begin{lemma} \label{Maxlemma}
   Let $A$ and $B$ be  $d$-dimensional positive definite matrices.  Let
              \begin{equation}
                F(x)\equiv \int_{{\mathbb R}^d} \exp{\left( -\langle y,Ay\rangle -\langle(x+y),B (x+y)\rangle       \right)} \, dy.
                \end{equation}    
Then $\ln(F(x))$ is quadratic  and concave in $x$  and maximal at $x=0$.           
          \end{lemma}
\begin{proof}
       By completing the square, one finds that
       \begin{equation}
       F(x) = \pi^{d/2}\left(\det(A+B)    \right)^{-1/2}\exp\left(\langle Bx, (A+B)^{-1} Bx\rangle -\langle x, B, x\rangle\right).
       \end{equation}
Since $A$ and $B$ are positive definite, $B\leq A+B$,  $(A+B)^{-1}\leq B^{-1}$, and so
$B(A+B)^{-1}B -B$ is negative definite. 
\end{proof}

Call $h(x)$ {\it Gaussian definite} provided that $h$ has an integral representation
      $$h(x) = \int_0^{\infty} \hat{h}(s) e^{-s x^2}ds,$$
with $\hat{h}(s)$ non-negative. More generally $h$ could be of the form
      $$h(x) = \int  e^{-\langle x,A_sx\rangle}d\mu(s),$$
with $\left\{A_s\right\}$ a family of positive definite matrices, $\mu$ a positive measure.
Then
\begin{lemma}
    Suppose that $h_1, h_2, ...h_n$ are Gaussian definite functions of $x\in {\mathbb R}^d$.  Then
    \begin{equation}
    F(x)=\int_{{\mathbb R}^d}\exp(-\langle y,Ay\rangle)\prod_{i = 1}^n h_i(x+y)\,dy
\end{equation}
is maximal at $x= 0$.  
\end{lemma}
\begin{proof} Write the product using the integral representations for the $h_i$'s and apply Fubini. \end{proof}

\begin{theorem}
Let $V$ be a Gaussian definite potential. For any measurable, non-negative function $f$,
\begin{equation}
\sup_{x \in \mathbb{R}^d}E^x\left[\exp\left(\int_0^T f(t)V(X_t)\,dt\right)\right] = 
E\left[\exp\left(\int_0^T f(t)V(X_t)\,dt\right)\right].
\end{equation}
\end{theorem}
\begin{proof}
Expanding the exponential as a power series  yields an infinite sum of integrals of terms of the form
\begin{gather}
E\left(V(X_{t_1})\ldots V(X_{t_n})\right) = 
C\int e^{-(y, Ay)}V(y_1 + x)\ldots V(y_n + x)\,dy
\end{gather}
for some $0 < t_1 < \ldots < t_n < T$, constant $C$, and positive matrix $A$. The result then follows from the previous lemma.
\end{proof}
In particular, we have that $|x|^{-\theta}$ is Gaussian-definite for $0 \leq \theta \leq 2$, hence the 
maximality at $x=0$ of expectations of the sort (\ref{appendix1.eq}) considered in this article.
\end{subsection}

\begin{subsection}{Note on the harmonic oscillator}
Here we provide a brief account of the 1-dimensional harmonic oscillator, and how our method can be modified to compute exactly the functional integral 
yielding its ground state energy. The result is a Cameron-Martin formula of which several proofs 
exist: see Cameron and Martin \cite{CM} (1945), and also \cite{KS1, RY}. The derivation we provide here is
simply an alternate route to their result.

       The harmonic oscillator Hamiltonian is given by $H= -\frac{1}{2}\frac{d^2}{dx^2} +\frac{\omega^2}{2}x^2$ with eigenvalues
 $E_n= (n+1/2)\omega$, $n= 0,1,...$. Its Feynman-Kac action is thus
             \begin{equation}\label{harmonic.05}
     S_T = -\frac{\omega^2}{2}\int _0^{T}X_t^2 dt.
     \end{equation}
Expanding this action via Clark-Ocone gives for example $E\left[ S_T  \right]  =   -\frac{\omega^2T^2}{4} $, which
is not log-linear,  growing as $T^2$ rather than $T$ and unlikely to be helpful in estimating the ground
state energy.    But, with the quadratic potential, one can solve explicitly a Hamilton-Jacobi-like equation, albeit
in integral form,  for
the stochastic integral integrand $\rho_T(s,\omega)$ given by the integral equation
          \begin{equation}\label{harmonic.1}
              S_T = \int_0^T\rho_T(s) dX_s -\frac{1}{2}\int_0^T \rho_T(s)^2 ds + f(T)
              \end{equation}
              so that 
            \begin{eqnarray}
                    E\left[ e^{S_T}\right] &=& e^{f(T)} E\left[ \exp{\left(\int_0^T\!\!\!\rho_T(s) dx(s) -\frac{1}{2}\int_0^T \!\!\!\rho_T(s)^2 ds \right)}\right]\nonumber\\
                    &=& e^{f(T)},
                    \end{eqnarray}
with the ground state energy $E_0=-\lim_{T\rightarrow\infty}\frac{f(T)}{T} $.                    
                  
  Writing Eqs.(\ref{harmonic.05},
\ref{harmonic.1}) in a Clark-Ocone expansion, we have that
  \begin{eqnarray}\label{harmonic.2}
         \lefteqn{ E\left[ D_s S_{T}|{\cal F}_s\right] = -\omega^2(T-s)X_s}
         \nonumber\\
         &=& \rho_T(s)- \int_s^T\!\!\!\, E\left[ \rho_T(t)D_s\rho_T(t)|{\cal F}_s \right]dt,
          \end{eqnarray}         
and then, having determined $\rho_T(s)$, we set
        \begin{equation}\label{harmonic.3}
           f(T)= E\left[S_T  \right]  +\frac{1}{2}  E\left[ \int_0^T\!\!\! \rho_T(s)^2 ds  \right].
      \end{equation}
Because of the simple form of the quadratic potential, we make the ansatz $\rho_T(s) \equiv r(s,T)X_s$, with  $r(s,T)$ deterministic, so
that 
(\ref{harmonic.2}) reduces to a deterministic integral equation
 \begin{eqnarray}\label{harmonic.2.5}
        -\omega^2(T-s)  &=& r(s,T)- \int_s^T\!\!\!\, r(t,T)^2dt
          \end{eqnarray}
with solution given by $r(s,T)= \omega \tanh{\omega(s-T)}$. (Note that the integral equation \eqref{harmonic.2.5} 
provides the boundary condition $r(T,T)= 0$.)  Then $f(T)$ is readily computed:
\begin{equation}\label{harmonic.4}
                f(T)= -\frac{1}{2}\ln{ (\cosh{{(\omega T))}}}.
                \end{equation}
From this one infers  the ground state energy  $E_0= \frac{\omega}{2}$ and, with some 
ingenuity, the higher eigenvalues $E_n$  with $n= 2, 4,....$ as well.

The above calculations amount to a derivation of the transformation 
\begin{equation}
   X_t \rightarrow X_t + \int_0^t \!\!\rho_T(s)ds = X_t + \omega \int_0^t\!\!\! \tanh{\omega(s-T)}X_s\,ds.               
 \end{equation}
This is a special case both of  the transformations constructed by Cameron and Martin in their derivation
 of Eq.(\ref{harmonic.4}) and for expectations in related time-inhomogeneous problems with 
 quadratic potentials.  They  used differential equation methods, however, rather than stochastic integrals.
\end{subsection}
\end{section}
\end{appendix}



\begin{thebibliography}{99}

\bibitem{BB1}R.D. Benguria, G.A. Bley, Exact asymptotic behavior of the Pekar-Tomasevich functional, J. Math. Phys. {\bf 52}, 052110 (2011)

\bibitem{BB2}R.D. Benguria, G.A. Bley, Improved results on the no-binding of bipolarons, J. Phys. A: Math. Theor. {\bf 45}, 045205 (2012)

\bibitem{CM} R.H. Cameron, W.T. Martin, Evaluation of various Wiener integrals by certain Sturm-Liouville differential equations, Bull. Amer. Math. Soc. {\bf 51}, 73-90, (1945)   

\bibitem{Case}K.M. Case, Singular Potentials, Phys. Rev. {\bf 80}, 5 (1950)

\bibitem{DV}M.D. Donsker, S.R.S. Varadhan, Asymptotics for the polaron, Comm. Pure Appl. Math., 36, 505-528 (1983)

\bibitem{EG}A.M. Essin, D.J. Griffiths, Quantum Mechanics of the $1/x^2$ potential, Am. J. Phys. {\bf 74}, 109 (2006)

\bibitem{RF} R.P. Feynman, Slow electrons in a polar crystal, Phys. Rev. {\bf 97} 660--665 (1958)

\bibitem{FLST}R.L. Frank, E.H. Lieb, R. Seiringer, L.E. Thomas, Stability and absence of binding for multi-polaron systems, 
Publication Math\'{e}matiques de l'IHES {\bf 113}, Number 1, 39-67 (2011)

\bibitem{HF} H. Fr\"{o}lich, Proc. Phys. Soc. {\bf A160}, 230 (1937); Electrons in lattice fields. Adv. in Phys. {\bf 3}, 325--361 (1954)  

\bibitem{GJ} J. Glimm and A. Jaffe, {\it Quantum Physics, a Functional Integral Point of View}, Springer-Verlag, 
New York, (1981). C.f. Sec. 9.1.

\bibitem{GHL} M. Gubinelli, F. Hiroshima, J. L\"{o}rinczi, Ultraviolet renormalization of the Nelson
Hamiltonian through functional integration, Journal of Functional Analysis, {\bf 267}, 3125-3153 (2014)

\bibitem{KS1} I. Karatzas, S. Shreve, {\it Brownian Motion and Stochastic Calculus}, 2nd ed., Springer, New York, (1991)

\bibitem{KS2}I. Karatzas, S. Shreve, {\it Methods of Mathematical Finance}, Springer, (2001)

\bibitem{LG1} J.-F. Le Gall,  Some properties of planar Brownian motion. \'{E}cole d'\'{E}t\'{e} de
 Probabilit\'{e}s  de Saint-Flour XX. Lecture Notes in Math. {\bf 1527} 111-235. Springer, Berlin, (1992)

\bibitem{LG2} J.-F. Le Gall, Exponential moments for the renormalized self-intersection local time
of planar Brownian motion. S\'{e}minaire de Probabilit\'{e}s XXVIII. Lecture Notes in Math.
{\bf 1583}, 172-180. Springer, Berlin, (1994)

\bibitem{L}E.H. Lieb, Existence and Uniqueness of the Minimizing Solution of Choquard's Nonlinear Equation, Studies in Applied 
Mathematics {\bf  57}, 93-105 (1977)

\bibitem{LL}E.H. Lieb, M. Loss, {\it Analysis}, 2nd ed., AMS, Providence, (2001)

\bibitem{LT}E.H. Lieb, L.E. Thomas, Exact ground state energy of the strong-coupling polaron, Commun. Math. Phys. {\bf 183}, 511-519 (1997). 
Erratum {\bf 188}, 499-500 (1997)

\bibitem{LY} E.H. Lieb, K. Yamazaki, Ground-state energy and effective mass of the polaron. Phys. Rev. {\bf 111}, 728-733 (1958)

\bibitem{MR} M.B.  Marcus, J. Rosen, Moment generating functions for local times of symmetric Markov processes and random walks. Probability in Banach spaces, {\bf 8},  (Brunswick, ME, 1991), 364-376, Progr. Probab., {\bf 30}, Birkh\"{a}user, Boston, MA, 1992. See also,  
Sample path properties of the local times of strongly symmetric Markov processes via Gaussian processes. Ann. Probab. {\bf 20} no. 4, 1603-1684   (1992)

\bibitem{M} S. J. Miyake, Strong coupling limit of the polaron ground state, J. Phys. Soc. Jpn., 38, 181-182 (1975)

\bibitem{N1} E. Nelson, Schr\"{o}dinger particles interacting with a quantized scalar field, in {\it Analysis in Function Space: Proceedings of a conference on the theory and application of 
  analysis in function space held at Endicott House in Dedham, Mass. June 9-13, 1963},  W.T. Martin and I. Segal, eds.,  87-121,  MIT Press, Cambridge, Mass, (1964)    
\bibitem{N2} E. Nelson, Interaction of nonrelativistic particles with a quantized scalar field, J. Math. Phys. {\bf 5}, 1190-1197 (1964)    

\bibitem{N} D. Nualart, {\it The Malliavin Calculus and related topics}, 2nd ed., Springer, Berlin,  (2006)

\bibitem{Simon2}M. Reed, B. Simon, {\it Methods of Modern Mathematical Physics II: Fourier Analysis, Self-Adjointness}, Academic Press, San Diego, (1975)

\bibitem{RY}D. Revuz, M. Yor, {\it Continuous Martingales and Brownian Motion}, 3rd ed., Springer, Berlin, (2010)

\bibitem{V} S.R.S. Varadhan, Appendix to Euclidean quantum field theory, by K. Symanzik, in {\it Local quantum field
theory}, R.Jost ed., Academic Press, New York, (1969)

\end{thebibliography}
\end{document}